\documentclass{dmtcs}
\usepackage{amssymb}
\usepackage{epstopdf}
\usepackage{delarray}
\usepackage{multirow}
\usepackage{stmaryrd}
\usepackage{amsmath}
\usepackage{wrapfig}
\usepackage{tikz}
\usepackage{color}
\newtheorem{definition}{Definition}
\newtheorem{notation}{Notation}

\newtheorem{lemma}{Lemma}
\newtheorem{theorem}{Theorem}
\newtheorem{corollary}{Corollary}
\newtheorem{proposition}{Proposition}

\author{K\'evin Perrot\addressmark{1}\thanks{\email{kevin.perrot@ens-lyon.fr}}
  \and Thi Ha Duong Phan\addressmark{2}\thanks{\email{phanhaduong@math.ac.vn}}
  \and Trung Van Pham\addressmark{2}\thanks{\email{pvtrung@math.ac.vn}}}
\title{On the set of Fixed Points of the Parallel Symmetric Sand Pile Model\thanks{This work is supported in part by the National Fundamental Research Programme in Natural Sciences of Vietnam, and the Complex System Institute of Lyon.}}
\address{\addressmark{1}LIP (UMR 5668 - CNRS - Universit\'e de Lyon - ENS de Lyon) - 46 all\'e d'Italie 69364 Lyon Cedex 7, France\\
  \addressmark{2}Institute of Mathematics, VAST - 18 Hoang Quoc Viet Road, Cau Giay,10307,  Hanoi, Vietnam}

\begin{document}
\maketitle
\begin{abstract}
  Sand Pile Models are discrete dynamical systems emphasizing the phenomenon of {\em Self-Organized Criticality}. From a configuration composed of a finite number of stacked grains, we apply on every possible positions (in parallel) two grain moving transition rules. The transition rules permit one grain to fall to its right or left (symmetric) neighboring column if the difference of height between those columns is larger than 2. The model is nondeterministic and grains always fall downward. We propose a study of the set of fixed points reachable in the Parallel Symmetric Sand Pile Model (PSSPM). Using a comparison with the Symmetric Sand Pile Model (SSPM) on which rules are applied once at each iteration, we get a continuity property. This property states that within PSSPM we can't reach every fixed points of SSPM, but a continuous subset according to the lexicographic order. Moreover we define a successor relation to browse exhaustively the sets of fixed points of those models.
\end{abstract}


\section{Introduction}

Sand Pile Models were introduced in 1988 (\cite{soc}) to highlight {\em Self-Organized Criticality} (SOC). SOC characterizes dynamical systems having critical attractors, {\em i.e.}, systems that evolve toward a stable state from which small perturbations have uncontrolled consequences on the system. This property is straightforward to figure out in the scope of sand pile models : consider a flat table on which we add grains one by one. After a moment, the amount of grains will look like a circular cone which base diameter will continue to grow as we add grains one by one. Some grain additions create avalanches, chain reactions involving numerous grain falls. Some avalanches stop quickly, others continue until they reach the table top. Now remark that whatever the size of the pile is, there will always be one more single grain addition which will increase the base diameter of the cone. So the tiniest possible perturbation --- one single grain addition --- can create an unbounded avalanche. This example illustrates the SOC of sand pile models.

There are many variants of sand pile models. All of them consider local grain moving transitions, applied in sequential or parallel mode (one rule application at each iteration or as many rule applications as possible at each iteration), starting from a finite number of stacked grains. The first model, introduced in \cite{spm}, considers one rule applied sequentially : if the difference of height between columns $i$ and $i+1$ is larger than two, then one grain falls from column $i$ to column $i+1$. The set of reachable configurations has a lattice structure and some other interesting properties, see \cite{order} and \cite{structure}. Furthermore its set of reachable configurations can be generated efficiently (see \cite{massa1}, \cite{massa2} and \cite{massa3}). Applying the rule in parallel on every possible column leads to a completely different description of the model, see \cite{pspm}. We can also add one more rule, symmetric to the previous one : if the difference of height between columns $i$ and $i-1$ is larger than two, then one grain can fall from column $i$ to column $i-1$. This leads to SSPM (symmetric sand pile model), studied in \cite{unimodal} and \cite{sspm}.

In \cite{fppt}, the authors studies PSSPM, the parallel variant of SSPM, and they proved that the form of fixed points of the two models are the same. In this paper, we investigate the set of all fixed points of PSSPM, taking into account their position. We provide a deterministic procedure to reach the extremal (leftmost and rightmost) fixed points of PSSPM according to the total lexicographic order, and prove that any fixed point between these two extremal fixed points reachable in SSPM is also reachable in PSSPM. We also define a successor relation $\triangleleft$ which gives a straightforward way of computing the set of fixed points of PSSPM.

In \cite{bspm} the authors suggest to add rules to get grains also moving forward and backward, to get closer to real life sand piles. \cite{survey} is a survey on sand pile models. An interesting generalization of sand pile models is sand automata, which are powerful enough to simulate cellular automata, see \cite{sa1}, \cite{sa2} and \cite{sa3}.

In this paper, $n$ is a given nonnegative integer.


\section{Parallel Symmetric Sand Pile Model}

In the theory of discrete dynamical systems, a model is defined by its set of configurations and its transition rule(s). We say that a configuration $b$ is {\em reachable} from a configuration $a$ if $b$ is obtained from $a$ by a sequence of transitions. In the scope of sand piles, we are interested in the set of configurations reachable from a finite number of stacked grains.

\begin{notation}
  A configuration $c$ is a sequence of nonnegative integers, with only finitely many positive values. We use an underlined number to denote the position 0 of a sequence. For example $c=(1,4,\underline{3},2,1)$ is the configuration such that $c_{-2}=1,c_{-1}=4,c_0=3,c_1=2,c_2=1$ and for all $i \notin \llbracket -2;2 \rrbracket, c_i=0$.
\end{notation}

We now give formal definitions of SSPM and PSSPM.

\begin{definition}
  SSPM is a discrete dynamical system defined by:
  \begin{itemize}
    \item Initial configuration: $(\underline{n})$.
    \item Local left vertical rule $\mathcal L$: $(\dots,a_{i-1},a_i,\dots) \to (\dots,a_{i-1}+1,a_i-1,\dots)$ if $a_{i-1}+2 \leq a_i$.
    \item Local right vertical rule $\mathcal R$: $(\dots,a_i,a_{i+1},\dots) \to (\dots,a_i-1,a_{i+1}+1,\dots)$ if $a_i \geq a_{i+1} +2$.
    \item Global rule: we apply once the $\mathcal L$ rule, or once the $\mathcal R$ rule.
  \end{itemize}
\end{definition}

SSPM is a non deterministic and sequential model. PSSPM is defined similarly with the rules applies in parallel on each column:

\begin{definition}
  PSSPM is a discrete dynamical dynamical system defined with the same initial configuration and local rules as SSPM, and the following global rule:
  \begin{itemize}
    \item Global rule: we apply $\mathcal L$ and $\mathcal R$ in parallel on every possible column. We apply at most one of the two rules on each column.
  \end{itemize}
\end{definition}

PSSPM is also a non deterministic model, for example from the initial configuration $(\underline{5})$ one has to choose whether applying $\mathcal L$ or $\mathcal R$ on column 0.

Once the model (SSPM or PSSPM) is fixed, we denote $a \to b$ when configuration $a$ reduces in one step to configuration $b$ according to the global transition rule. $\to^*$ denotes the transitive closure of $\to$. We formally define the sets of {\em reachable} configurations as:

\begin{notation}
  SSPM($n$)$=\bigcup \{a | (\underline{n}) \to^* a\}$ is the set of reachable configurations from the initial configuration $(\underline{n})$ by applying SSPM rules.
  
  PSSPM($n$)$=\bigcup \{a | (\underline{n}) \to^* a\}$ is the set of reachable configurations from the initial configuration $(\underline{n})$ by applying PSSPM rules.
  
  SSPM = $\bigcup \limits_{n \in \mathbb{N}}$ SSPM($n$) and PSSPM = $\bigcup \limits_{n \in \mathbb{N}}$ PSSPM($n$).
\end{notation}

In both models, one can note that any configuration $c$ reachable from the initial configuration $(\underline{n})$ verifies $\sum \limits_{i} c_i = n$ and  for some $j$, $\dots \leq a_{j-2} \leq a_{j-1} \leq a_j \geq a_{j+1} \geq a_{j+2} \geq ...$. This last observation leads to the fact that within PSSPM, there is at most one column $j$ on which a choice between $\mathcal L$ and $\mathcal R$ happens (such a column $j$ must verify $a_{j-1} < a_j$ and $a_j > a_{j+1}$).

On figure \ref{fig:psspm5} we present in PSSPM the complete set of reachable configurations from $(\underline{5})$. A reachable configuration from which no transition can be applied is a {\em fixed point}.

\begin{figure}[!h]
  \begin{center}
  \begin{tikzpicture}
    \def\px0{0}
    \def\py0{0}
    \foreach \y in {0,0.3,...,1.5}{
      \draw (\px0,\py0 + \y) rectangle ++ (0.3,0.3);
    }
    \draw[line width=2pt] (\px0,\py0) -- ++ (.3,0);   
    
    \def\px1{-1}
    \def\py1{-1.5}
    \foreach \y in {0,0.3,...,1.2}{
      \draw (\px1,\py1 + \y) rectangle ++ (0.3,0.3);
    }
    \draw (\px1 - 0.3, \py1) rectangle ++ (0.3,0.3);
    \draw[line width=2pt] (\px1,\py1) -- ++ (.3,0);
    
    \def\px1{1}
    \def\py1{-1.5}
    \foreach \y in {0,0.3,...,1.2}{
      \draw (\px1,\py1 + \y) rectangle ++ (0.3,0.3);
    }
    \draw (\px1 + 0.3, \py1) rectangle ++ (0.3,0.3);
    \draw[line width=2pt] (\px1,\py1) -- ++ (.3,0);
    
    \def\px1{-2}
    \def\py1{-3}
    \foreach \y in {0,0.3,...,0.9}{
      \draw (\px1,\py1 + \y) rectangle ++ (0.3,0.3);
    }
    \draw (\px1 - 0.3, \py1) rectangle ++ (0.3,0.3);
    \draw (\px1 - 0.3, \py1 + 0.3) rectangle ++ (0.3,0.3);
    \draw[line width=2pt] (\px1,\py1) -- ++ (.3,0);
    
    \def\px1{2}
    \def\py1{-3}
    \foreach \y in {0,0.3,...,0.9}{
      \draw (\px1,\py1 + \y) rectangle ++ (0.3,0.3);
    }
    \draw (\px1 + 0.3, \py1) rectangle ++ (0.3,0.3);
    \draw (\px1 + 0.3, \py1 + 0.3) rectangle ++ (0.3,0.3);
    \draw[line width=2pt] (\px1,\py1) -- ++ (.3,0);
    
    \def\px1{0}
    \def\py1{-3}
    \foreach \y in {0,0.3,...,0.9}{
      \draw (\px1,\py1 + \y) rectangle ++ (0.3,0.3);
    }
    \draw (\px1 - 0.3, \py1) rectangle ++ (0.3,0.3);
    \draw (\px1 + 0.3, \py1) rectangle ++ (0.3,0.3);
    \draw[line width=2pt] (\px1,\py1) -- ++ (.3,0);
    
    \def\px1{-1}
    \def\py1{-4.2}
    \draw (\px1 - 0.3, \py1) rectangle ++ (0.3,0.3);
    \draw (\px1 - 0.3, \py1 + 0.3) rectangle ++ (0.3,0.3);
    \draw (\px1, \py1) rectangle ++ (0.3,0.3);
    \draw (\px1, \py1 + 0.3) rectangle ++ (0.3,0.3);
    \draw (\px1 + 0.3, \py1) rectangle ++ (0.3,0.3);
    \draw[line width=2pt] (\px1,\py1) -- ++ (.3,0);
    
    \def\px1{1}
    \def\py1{-4.2}
    \draw (\px1 - 0.3, \py1) rectangle ++ (0.3,0.3);
    \draw (\px1 + 0.3, \py1 + 0.3) rectangle ++ (0.3,0.3);
    \draw (\px1, \py1) rectangle ++ (0.3,0.3);
    \draw (\px1, \py1 + 0.3) rectangle ++ (0.3,0.3);
    \draw (\px1 + 0.3, \py1) rectangle ++ (0.3,0.3);
    \draw[line width=2pt] (\px1,\py1) -- ++ (.3,0);
    
    \def\px1{-2}
    \def\py1{-5.4}
    \draw (\px1 - 0.3, \py1) rectangle ++ (0.3,0.3);
    \draw (\px1 - 0.6, \py1) rectangle ++ (0.3,0.3);
    \draw (\px1, \py1) rectangle ++ (0.3,0.3);
    \draw (\px1, \py1 + 0.3) rectangle ++ (0.3,0.3);
    \draw (\px1 + 0.3, \py1) rectangle ++ (0.3,0.3);
    \draw[line width=2pt] (\px1,\py1) -- ++ (.3,0);
    
    \def\px1{2}
    \def\py1{-5.4}
    \draw (\px1 - 0.3, \py1) rectangle ++ (0.3,0.3);
    \draw (\px1 + 0.6, \py1) rectangle ++ (0.3,0.3);
    \draw (\px1, \py1) rectangle ++ (0.3,0.3);
    \draw (\px1, \py1 + 0.3) rectangle ++ (0.3,0.3);
    \draw (\px1 + 0.3, \py1) rectangle ++ (0.3,0.3);
    \draw[line width=2pt] (\px1,\py1) -- ++ (.3,0);
    
    \draw [->] (-0.2,-0.2) -- node [above] {\scriptsize $\mathcal L$} (-0.5,-0.4);
    \draw [->] (0.5,-0.2) -- node [above] {\scriptsize $\mathcal R$} (0.8,-0.4);
    
    \draw [->] (-1.2,-1.7) -- node [left] {\scriptsize $\mathcal L$} (-1.5,-2.1);
    \draw [->] (-0.5,-1.7) -- node [right] {\scriptsize $\mathcal R$} (-0.2,-2.1);
    \draw [->] (0.8,-1.7) -- node [left] {\scriptsize $\mathcal L$} (0.5,-2.1);
    \draw [->] (1.5,-1.7) -- node [right] {\scriptsize $\mathcal R$} (1.8,-2.1);
    
    \draw [->] (-1.85,-3.2) -- (-1.85,-4.6);
    \draw [->] (2.15,-3.2) -- (2.15,-4.6);
    
    \draw [->] (-0.2,-3.2) -- node [left] {\scriptsize $\mathcal L$} (-0.5,-3.6);
    \draw [->] (0.5,-3.2) -- node [right] {\scriptsize $\mathcal R$} (0.8,-3.6);
    
    \draw [->] (-1.2,-4.4) -- (-1.5,-4.7);
    \draw [->] (1.5,-4.4) -- (1.8,-4.7);
  \end{tikzpicture}
  \caption{The set of reachable configurations in PSSPM starting from the initial configuration $(\underline{5})$. A bold line denotes column 0. Edges are labelled according to the choice $\mathcal L$ or $\mathcal R$, whenever there is one. Two fixed points are reachable from $(\underline{5})$: $(1,1,\underline{2},1)$ and $(1,\underline{2},1,1)$.}
  \label{fig:psspm5} 
  \end{center}
\end{figure}
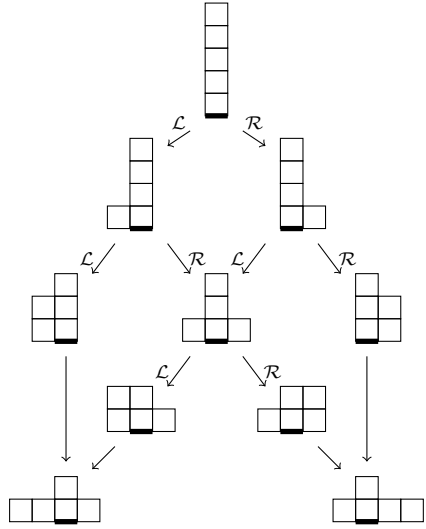

A trivial --- nevertheless motivating --- result is that the set of reachable configurations in PSSPM is a subset of reachable configurations in SSPM:

\begin{proposition}\label{prop:psspm}
  PSSPM $\subsetneq$ SSPM.
\end{proposition}

\begin{proof} 
  PSSPM $\subseteq$ SSPM is obvious.
  Let us show that PSSPM(5) $\subsetneq$ SSPM(5) which leads to the result. Using SSPM global transition rule, $(\underline{5}) \to (\underline{4},1) \to (\underline{3},2) \to (\underline{3},1,1) \to (\underline{2},2,1) \to (1,\underline{1},2,1)$, so $(1,\underline{1},2,1) \in SSPM(5)$. On figure \ref{fig:psspm5} we can see that using PSSPM parallel rule application, $(1,\underline{1},2,1) \notin PSSPM(5)$.
\end{proof} 

The set of fixed points of PSSPM is strictly included in the set of fixed points of SSPM (note that it does not hold in the one sided case, where SPM and PSPM have exactly the same fixed points). The following section concentrates on the properties of former compared to latter.


\section{Fixed points of PSSPM}

We propose a study of the set of fixed points of PSSPM. We give a deterministic procedure to reach the rightmost and leftmost fixed points of PSSPM($n$), respectively corresponding to the smallest and greatest configurations according to the lexicographic order. We also prove that every fixed point of SSPM($n$) between the smallest and greatest fixed point of PSSPM($n$) are reachable in PSSPM($n$). As a consequence, the set of PSSPM($n$) fixed points inherits a kind of continuity property.

The {\em transition diagram} of PSSPM($n$) is the edge-labeled directed multigraph $G_n=(V_n,E_n)$ where $V_n=$ PSSPM($n$) is the set of reachable configurations from the initial configuration $(\underline{n})$, and $E_n \subseteq V_n \times V_n \times \{\mathcal L, \mathcal R\}$ such that $(a,b,\alpha) \in E_n$ if and only if $a \to b$ according to PSSPM rules where we choose (recall that there is at most one choice) to apply the $\alpha$ rule (when there is no choice from $a$ to $b$, both $(a,b,\mathcal L)$ and $(a,b,\mathcal R)$ belong to $E_n$). Figure \ref{fig:psspm5} is the transition diagram of PSSPM(5), where multiple edges are replaced by a single unlabeled edge.

From a configuration $a$, we consider the two configurations obtained according to the choices $\mathcal L$ and $\mathcal R$. Then, we let those two configurations evolve using the same choice at each step. We obtain two sequences of configurations, and we will see that they stay very close, in other words they represent very similar paths within the transition diagram. We introduce a formal notation, $\mathcal L(a)$, standing for the configuration obtained by choosing the top grain to fall to the left if possible (if it is not possible, the top grain falls to the right):

\begin{notation}
  Let $a$ be a configuration such that $a \in V_n$ for some fixed integer $n$. $\mathcal L(a)$ is the configuration defined as:
  \begin{enumerate}
    \item if $\exists~ b$ such that $(a,b,\mathcal L) \in E_n$ then $\mathcal L(a)=b$,
    \item else if $\exists~ b$ such that $(a,b,\mathcal R) \in E_n$ then $\mathcal L(a)=b$,
    \item else $\mathcal L(a)=a$.
  \end{enumerate}
  $\mathcal R(a)$ is defined similarly.
  
  Let $\omega=\omega_1\dots\omega_k$ be a word over the alphabet $\{\mathcal L,\mathcal R\}$, $\omega(a)$ is the configuration defined inductively as $\omega(a)=\omega_2\dots\omega_k(\omega_1(a))$.
\end{notation}

The idea will be to consider a configuration $a$ of PSSPM($n$) for a fixed integer $n$ and the two configurations $\mathcal R(a)$ and $\mathcal L(a)$. Those two configurations are intuitively similar each other. Then we will see that for every word $\omega$ over the alphabet $\{\mathcal L,\mathcal R\}$, the configurations $\omega(\mathcal R(a))$ and $\omega(\mathcal L(a))$ are also similar according to the relation $\overset{*}{\triangleleft}$ defined below. This is the key argument of our study, stated in Proposition \ref{lemma:main}. Finally, we use known results about SSPM and further developments to show that when we reach fixed points, the configurations are very similar (see $\triangleleft$ defined below). This leads to Theorem \ref{theorem:continuous}, relating the set of fixed points reachable in PSSPM($n$) to that reachable in SSPM($n$).

\begin{definition}
  Let $\Delta(a,\!b)$ be the sequence of differences between configurations $a$ and $b$, $\Delta_i(a,\!b)=a_i-b_i$.\\
  We define a notion of similarity or closeness between configurations, denoted by the following relations:
  \[
    \begin{array}{rcl}
      a \triangleleft b & \iff & \Delta(a,b) \in 0^*-\!\!10^*10^*\\
      a \overset{*}{\triangleleft} b & \iff & \Delta(a,b) \in (0^*-\!\!10^*10^*)^*
    \end{array}
  \]
  where $-\!1$ is a minus one value. As a convention $\epsilon=0^\omega$, so that $a=b$ implies $a \overset{*}{\triangleleft} b$.
\end{definition}

 The reader should note that $\overset{*}{\triangleleft}$ is not the reflexive transitive closure of $\triangleleft$, it is just a kind of {\em non-strict} variant of $\triangleleft$.

The following lemma states the similarity of the configurations obtained when we follow very close paths in the transition diagram of PSSPM($n$). The weak relation $\overset{*}{\triangleleft}$ is used to compare obtained configurations all along the evolution toward fixed points. We will see in Proposition \ref{lemma:psspm} that the relation between fixed points can be strengthened into $\triangleleft$.

\begin{proposition}\label{lemma:main}
  Let $a \in$ PSSPM($n$). For all $\omega \in \{\mathcal L,\mathcal R\}^*,$
  $$\omega(\mathcal R(a)) \overset{*}{\triangleleft} \omega(\mathcal L(a))$$
\end{proposition}

We first present a technical lemma used to avoid some impossible cases in the proof of Proposition \ref{lemma:main}.

\begin{lemma}[technical]\label{lemma:technical}
  Consider a sequence in PSSPM($n$).
  $$c^1 \to c^2 \to \dots \to c^k$$
  If there exists a column $i$ such that
    \begin{enumerate}
      \item $i$ remains one of the highest columns
      {\em i.e.}, $\forall~ 1 \leq t \leq k,~ c^t_i = \max \limits_j c^t_j$
      \item $c^1_i \leq c^1_{i+1} + 2$ (resp. $c^1_{i-1} + 2 \geq c^1_i$)
    \end{enumerate}
  Then $\forall~ 1 \leq t \leq k,~ c^t_i \leq c^t_{i+1} + 2$ (resp. $c^t_{i-1} + 2 \geq c^t_i$).
\end{lemma}
\begin{proof}
  We proceed by induction on the iterations. The base case is verified according to the second hypothesis. The top column $i$ can't receive any grain during the iterations under consideration so the height difference with column $i+1$ (resp. $i-1$) can only be increased by 1 if $i$ doesn't lose a grain and $i+1$ (resp. $i-1$) loses a grain. In any other case the height difference doesn't increase. So if the height difference is at most 1, then it can't be increased to a difference greater than 2. If the height difference is 2, then column $i$ loses a grain so the height difference doesn't increase.
\end{proof}

\begin{proof}[of Proposition \ref{lemma:main}]
  We proceed by induction on the length of $\omega$. The base case is obvious : either there is no choice from $a$ to $\mathcal L(a)$ and $\mathcal R(a)$, and hence $\mathcal L(a) = \mathcal R(a)$ implies $\mathcal R(a) \overset{*}{\triangleleft} \mathcal L(a)$ or there is a choice on column $i$ and hence
  \begin{itemize}
    \item $\mathcal R(a)_{i-1} = \mathcal L(a)_{i-1} - 1$
    \item $\mathcal R(a)_i = \mathcal L(a)_i$
    \item $\mathcal R(a)_{i+1} = \mathcal L(a)_{i+1} + 1$
    \item $\forall j \notin \{i-1,i,i+1\},~ \mathcal R(a)_j = \mathcal L(a)_j$
  \end{itemize}
  so $\mathcal R(a) \overset{*}{\triangleleft} \mathcal L(a)$.
  By induction hypothesis, we are considering two configurations $b=\omega_1\dots\omega_{k-1}(\mathcal R(a))$ and $c=\omega_1\dots\omega_{k-1}(\mathcal L(a))$ such that $b \overset{*}{\triangleleft} c$ and we will now prove that $\omega_k(b) \overset{*}{\triangleleft} \omega_k(c)$.
  
  For the sake of clarity, we denote $d$ (resp. $e$) the configuration such that $b \overset{\omega_k}{\to} d$ (resp. $c \overset{\omega_k}{\to} e$).
    
  We do an induction on the columns, and construct $\Delta(d,e)$ from our knowledge on $\Delta(b,c)$, from left to right according to the behavior of each column $i$ in $b$ and $c$. Considering the rule application on column $i$ of a configuration $h$ gives us three informations:
  \begin{itemize}
    \item does column $i-1$ receive a grain from its right neighbor, denoted $\overset{\leftarrow}{h}_{i-1} \in \{0,1\}$;
    \item does column $i$ give a grain to one of its neighbors, denoted $\overline{h}_i \in \{0,1\}$;
    \item does column $i+1$ receive a grain from its left neighbor, denoted $\overset{\rightarrow}{h}_{i+1} \in \{0,1\}$.
  \end{itemize}
  
  In order to conclude, we will use the fact that
  $$\text{for all } j,~ \Delta_j(d,e)=\Delta_j(b,c) + (\overset{\leftarrow}{b}_j - \overset{\leftarrow}{c}_j) - (\overline{b}_j - \overline{c}_j) + (\overset{\rightarrow}{b}_j - \overset{\rightarrow}{c}_j)$$

  At each step of the induction, we will "update" $\Delta(b,c)$ with the three informations we get and see that it has always the form $(0^*-\!\!10^*10^*)^*$. We denote $\Delta^i(b,c)$ the sequence $\Delta(b,c)$ updated up to index $i$, defined at each index $j$ as
  $$\Delta^i_j(b,c)=\begin{array}\{{ll}. \Delta_j(b,c) + (\overset{\leftarrow}{b}_j - \overset{\leftarrow}{c}_j) - (\overline{b}_j - \overline{c}_j) + (\overset{\rightarrow}{b}_j - \overset{\rightarrow}{c}_j) = \Delta_j(d,e) & \text{if } j<i \\ \Delta_j(b,c) - (\overline{b}_i - \overline{c}_i) + (\overset{\rightarrow}{b}_i - \overset{\rightarrow}{c}_i) & \text{if } j=i \\ \Delta_j(b,c) + (\overset{\rightarrow}{b}_i - \overset{\rightarrow}{c}_i) & \text{if } j=i+1 \\ \Delta_j(b,c) & \text{if } j>i+1 \end{array}$$
  
  For initialization, there obviously exists an index $s$ such that for all $j \leq s$, there is no grain and hence no rule application on $j$ both in $b$ and $c$. Therefore $\Delta^s(b,c)=\Delta(b,c) \in (0^*-\!\!10^*10^*)^*$.
  
  Let us now eventually prove that for any $i$, $\Delta^{i-1}(b,c) \in (0^*-\!\!10^*10^*)^*$ implies $\Delta^i(b,c) \in (0^*-\!\!10^*10^*)^*$. This will complete the proof of the lemma, since there exists an index $t$ such that for all $j \geq t$, there is no grain and hence no rule application on $j$ both in $b$ and $c$. Therefore $\Delta^t(b,c)=\Delta(d,e)$.
  
  We prove that $\Delta^{i-1}(b,c) \in (0^*-\!\!10^*10^*)^*$ implies $\Delta^i(b,c) \in (0^*-\!\!10^*10^*)^*$ in three stages: left part, central part and right part. The central part is the set of columns where we apply different local rules in $b$ and $c$ (we will see that there is at most one column in the central part). The left (resp. right) part is the set of columns where grains can only fall to the left (resp. right) both in $b$ and $c$. The proofs for the left and right parts are symmetric. The central part is more involved and uses lemma \ref{lemma:technical}.
  \begin{itemize}
    \item left part.\\
    We consider an index $i$ which may be fired to the left or not fired. Since it is not fired to the right, $\overset{\to}{b}_{i+1} - \overset{\to}{c}_{i+1} = 0$ and every index in this part verifies that $\Delta^{i-1}_i(b,c)=\Delta_i(b,c)$. There are 4 cases, some of them are symmetric:
    \begin{itemize}
      \item $i$ fired to the left in both $b$ and $c$, then $\Delta^i(b,c)=\Delta^{i-1}(b,c)$.
      \item $i$ not fired in both $b$ and $c$, again $\Delta^i(b,c)=\Delta^{i-1}(b,c)$.
      \item $i$ fired to the left in $b$, not fired in $c$. Then we have the following changes in $\Delta^i(b,c)$:
      $$\left\{\begin{array}{l} \overset{\leftarrow}{b}_{i-1} - \overset{\leftarrow}{c}_{i-1} = 1\\ \overline{b}_i - \overline{c}_i = 1\\ \overset{\to}{b}_{i+1} - \overset{\to}{c}_{i+1} = 0 \end{array}\right. \text{ hence } \left\{\begin{array}{l} \Delta^i_{i-1}(b,c) = \Delta^{i-1}_{i-1}(b,c)+1\vspace{4pt}\\ \Delta^i_i(b,c) = \Delta^{i-1}_i(b,c)-1\vspace{4pt}\\ \text{elsewhere there is no change} \end{array}\right.$$
      but the rule application on $i$ involves that $b_{i-1} +2 \leq b_i$ and $c_{i-1} +2 > c_i$. There are 3 different cases according to the values of $(b_i - b_{i-1})$, $\Delta_{i-1}(b,c)$ and $\Delta_i(b,c)$: (for any other set of values we haven't  $i$ fired to the left in $b$ and $i$ not fired in $c$)
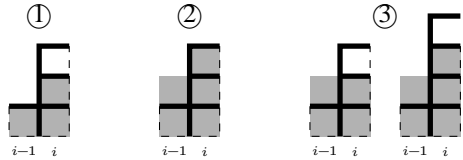
      \begin{center}
        \begin{tikzpicture}[scale=.4]
  \node at (1,4) {\textcircled{1}};
  \fill[fill=black!30] (0,0) rectangle ++ (2,1);
  \fill[fill=black!30] (1,1) rectangle ++ (1,1);
  \draw[line width=2pt] (0,1) -- ++ (2,0);
  \draw[line width=2pt] (1,2) -- ++ (1,0);
  \draw[line width=2pt] (2,3) -- ++ (-1,0) -- ++ (0,-3);
  \draw[dashed] (0,1) -- ++ (0,-1) -- ++ (2,0) -- ++ (0,3);
  \node at (1.5,-.5) {\tiny $i$};
  \node at (.5,-.5) {\tiny $i\!\!-\!\!1$};
  \node at (6,4) {\textcircled{2}};
  \fill[fill=black!30] (5,0) rectangle ++ (2,2);
  \fill[fill=black!30] (6,2) rectangle ++ (1,1);
  \draw[line width=2pt] (5,1) -- ++ (2,0);
  \draw[line width=2pt] (6,2) -- ++ (1,0);
  \draw[line width=2pt] (7,3) -- ++ (-1,0) -- ++ (0,-3);
  \draw[dashed] (5,1) -- ++ (0,-1) -- ++ (2,0) -- ++ (0,3);
  \node at (6.5,-.5) {\tiny $i$};
  \node at (5.5,-.5) {\tiny $i\!\!-\!\!1$};
  \node at (12.5,4) {\textcircled{3}};
  \fill[fill=black!30] (10,0) rectangle ++ (2,2);
  \draw[line width=2pt] (10,1) -- ++ (2,0);
  \draw[line width=2pt] (11,2) -- ++ (1,0);
  \draw[line width=2pt] (12,3) -- ++ (-1,0) -- ++ (0,-3);
  \draw[dashed] (10,1) -- ++ (0,-1) -- ++ (2,0) -- ++ (0,3);
  \node at (11.5,-.5) {\tiny $i$};
  \node at (10.5,-.5) {\tiny $i\!\!-\!\!1$};
  \fill[fill=black!30] (13,0) rectangle ++ (2,2);
  \fill[fill=black!30] (14,2) rectangle ++ (1,1);
  \draw[line width=2pt] (13,1) -- ++ (2,0);
  \draw[line width=2pt] (14,2) -- ++ (1,0);
  \draw[line width=2pt] (14,3) -- ++ (1,0);
  \draw[line width=2pt] (15,4) -- ++ (-1,0) -- ++ (0,-4);
  \draw[dashed] (13,1) -- ++ (0,-1) -- ++ (2,0) -- ++ (0,3);
  \node at (14.5,-.5) {\tiny $i$};
  \node at (13.5,-.5) {\tiny $i\!\!-\!\!1$};
\end{tikzpicture}\\
        $b$ is pictured with bold lines, $c$ is pictured in grey. If the difference of height between $i-1$ and $i$ is greater than 3 in $b$ then it is greater or equal to 2 in $c$. We recall that $b \overset{*}{\triangleleft} c$.
      \end{center}
      \begin{enumerate}
        \item[\textcircled{1}] $\Delta_{i-1}(b,c)=0$ and $\Delta_i(b,c)=1$.\\
        By induction hypothesis $\Delta^{i-1}(b,c) \in (0^*-\!10^*10^*)^*$, so we can deduce from the equality $\Delta^{i-1}_i(b,c)=\Delta_i(b,c)$ that $\Delta^{i-1}(b,c)$ around index $i$ is
        $$(\dots,-1,\dots,\underset{i}{1},\dots,-1,\dots)$$
        where the right $-1$ may not exist. Therefore, after applying the changes (adding 1 at index $i-1$ and subtracting 1 at index $i$) we still have $\Delta^i(b,c) \in (0^*-\!\!10^*10^*)^*$.
        \item[\textcircled{2}] $\Delta_{i-1}(b,c)=-1$ and $\Delta_i(b,c)=0$.\\
        By induction hypothesis $\Delta^{i-1}(b,c) \in (0^*-\!10^*10^*)^*$, and we also need that $\Delta^{i-2}(b,c) \in (0^*-\!\!10^*10^*)^*$ which is clear according to the base case. We can deduce from the equalities $\Delta^{i-2}_{i-1}(b,c)=\Delta_{i-1}(b,c)$ and for the same reason $\Delta^{i-2}_i(b,c)=\Delta_i(b,c)$ that $\Delta^{i-2}(b,c)$ around index $i-1$ is
        $$(\dots,1,\dots,\underset{i-1}{-1},\underset{i}{0},\dots,1,\dots)$$
        where the left $1$ may not exist. The part on the right of $i-1$ is not altered by the induction step from $i-2$ to $i-1$, therefore $\Delta^{i-1}(b,c)$ around index $i$ is
        $$(\dots,-1,\dots,\underset{i}{0},\dots,1,\dots)$$
        (it can't be equal to $0^\omega$ for the right 1 is still there). Therefore, after applying the changes (adding 1 at index $i-1$ and subtracting 1 at index $i$) we still have $\Delta^i(b,c) \in$ \mbox{$(0^*-\!\!10^*10^*)^*$}.
        \item[\textcircled{3}] $\Delta_{i-1}(b,c)=-1$ and $\Delta_i(b,c)=1$.\\
        The argument is the same as in the case \textcircled{1}.
      \end{enumerate}
      \item $i$ not fired in $b$, fired to the left in $c$. This case is symmetric to the previous one.
    \end{itemize}
    \item central part.\\
    Let us first prove by contradiction that there is at most one column which is fired using different local rules in $b$ and $c$. We name $u$ and $v$ ($u<v$) the two columns. There are two cases:
    \begin{itemize}
      \item In $b$, $u$ fires to the left and $v$ fires to the right. Then in $c$, $u$ fires to the right and $v$ fires to the left. This is impossible since $c$ is an increasing then decreasing sequence.
      \item In $b$, both $u$ and $v$ fires to the left. Then the height difference between $b_{u-1}$ and $b_v$ is at least 4. Since $b \overset{*}{\triangleleft} c$ the differences between $b$ and $c$ are at most 1 which makes impossible the case where $c_u - c_{v+1} \geq 2$ (necessary condition for $u$ to fire to the right in $c$).\\
    \end{itemize}
    We now consider the influence of the index $i$ where $b$ and $c$ have opposite behaviors. Let us take $\omega_k=\mathcal L$ and consider that $i$ is fired to the left in $b$ and to the right in $c$ (other cases are symmetric). We have the following changes in $\Delta^i(b,c)$:
    $$\left\{\begin{array}{l} \overset{\leftarrow}{b}_{i-1} - \overset{\leftarrow}{c}_{i-1} = 1\\ \overline{b}_i - \overline{c}_i = 0\\ \overset{\to}{b}_{i+1} - \overset{\to}{c}_{i+1} = -1 \end{array}\right. \text{ hence } \left\{\begin{array}{l} \Delta^i_{i-1}(b,c) = \Delta^{i-1}_{i-1}(b,c)+1\vspace{4pt}\\ \Delta^i_i(b,c) = \Delta^{i-1}_i(b,c)\vspace{4pt}\\ \Delta^i_{i+1}(b,c) = \Delta^{i-1}_{i+1}(b,c)-1\vspace{4pt}\\ \text{elsewhere there is no change} \end{array}\right.$$
    but the rule application on $i$ involves that $b_{i-1} +2 \leq b_i$, $c_{i-1}+2 > c_i$ (which prevents index $i$ in $c$ to follow the choice $\mathcal L$) and $c_i \geq c_{i+1}+2$. There are 3 cases which can be pictured exactly as in the left part.
    \begin{enumerate}
      \item[\textcircled{1}] $\Delta_{i-1}(b,c)=0$ and $\Delta_i(b,c)=1$.\\
      In this case, $b_{i+1} \leq c_{i+1}$. Since column $i$ in $c$ is fired to the left, $c_i \geq c_{i+1} + 2$, hence $b_i \geq b_{i+1} + 3$ because there is one more grain at $i$ in $b$.\\
      Also, $\Delta_i(b,c) \neq 0$ so there is one iteration during which a firing of index $i$ has been performed in an ancestor of $c$ and not in the corresponding ancestor of $b$ (in which the height difference between $i$ and $i+1$ was lesser than 2), or there is one iteration during which index $i$ received a grain in an ancestor of $b$ but not in the corresponding ancestor of $c$ (in this case, $i$ became and remains the highest column in the chain leading to $b$ and there exist an iteration where $i$ is not fired, so that it became the only highest, hence the height difference between $i$ and $i+1$ was lesser than 2).\\
      The conditions of lemma \ref{lemma:technical} are verified and $b_i \geq b_{i+1} + 3$, this case is impossible.
      \item[\textcircled{2}] $\Delta_{i-1}(b,c)=-1$ and $\Delta_i(b,c)=0$.\\
      By induction hypothesis $\Delta^{i-1}(b,c) \in (0^*-\!\!10^*10^*)^*$, so we can deduce from the fact that $i$ can't receive any grain (it is obviously one of the top columns of $b$ and $c$) that $\Delta^{i-1}_i(b,c)=\Delta_i(b,c)$. Moreover, $\Delta^{i-1}_j(b,c)$ for $j>i$ is still equal to $\Delta_j(b,c)$. Let us recall that $b \overset{*}{\triangleleft} c$ and $\Delta_{i-1}(b,c)=-1$. As a consequence, $\Delta^{i-1}(b,c)$ around index $i$ is
      $$(\dots,-1,\dots,\underset{i}{0},\dots,1,\dots)$$
      Therefore, after applying the changes (adding 1 at index $i-1$, subtracting 1 at index $i+1$), we still have $\Delta^i(b,c) \in (0^*-\!\!10^*10^*)^*$.
      \item[\textcircled{3}] $\Delta_{i-1}(b,c)=-1$ and $\Delta_i(b,c)=1$.\\
      For the same reason as above, we prove using lemma \ref{lemma:technical} that this case is impossible.
    \end{enumerate}
    \item right part.\\
    This part is symmetric to the left part.
  \end{itemize}
  We proved that $\Delta^{i-1}(b,c) \in (0^*-\!\!10^*10^*)^*$ implies $\Delta^i(b,c) \in (0^*-\!\!10^*10^*)^*$, which concludes the proof that $d \overset{*}{\triangleleft} e$, which in turn completes the proof of this lemma.
\end{proof}

This lemma states that trying to follow the same transitions conserves the relation $\overset{*}{\triangleleft}$. It provides a deterministic procedure to reach the extremal fixed points of PSSPM(n):

\begin{notation}
  We use the symbols $\leq_{lex}$ and $\geq_{lex}$ to denote the lexicographic order over configurations. Note that $a \overset{*}{\triangleleft} b \Rightarrow a \leq_{lex} b$ and $a \triangleleft b \Rightarrow a <_{lex} b$.
\end{notation}

\begin{corollary}
  The maximal --- leftmost --- (resp. minimal --- rightmost ---) fixed point of PSSPM($n$) according to the lexicographic order is reached when one chooses at every step the $\mathcal L$ rule (resp. $\mathcal R$ rule).
\end{corollary}
\begin{proof}
  By induction on Proposition \ref{lemma:main} and since $a \overset{*}{\triangleleft} b \Rightarrow a\leq_{lex} b$, we have for all $k \in \mathbb N$ and all $w \in \{ \mathcal L, \mathcal R \}^k$ that $\mathcal L^k((\underline{n})) \leq w((\underline{n}))$.
\end{proof}

We will now see how the relation $\overset{*}{\triangleleft}$, used strictly, allows one to browse exhaustively the set of reachable fixed points of SSPM($n$) and PSSPM($n$).

\begin{proposition}\label{lemma:psspm}
For all fixed $a$ of PSSPM($n$) except its leftmost (maximal according to $\leq_{lex}$), there exists a unique fixed point $b$ of PSSPM($n$) such that $a \triangleleft b$.
\end{proposition}
\begin{proof} 
  There exists a word $u$ such that $u((\underline{n}))=a$ and from Proposition \ref{lemma:main}, since $a$ is not the greatest fixed point of PSSPM($n$), by incrementally changing letters $\mathcal R$ into $\mathcal L$ in $u$ until reaching a configuration different from $a$, we will eventually find a configuration $b$ such that $a \overset{*}{\triangleleft} b$ and $a \neq b$.
  
  Let us now prove that $a \triangleleft b$ and that there is no other fixed point $c$ such that $a \triangleleft c$. By a result from \cite{sspm} and \cite{unimodal} a fixed point of SSPM (hence of PSSPM) can be cut into two parts which are fixed points of SPM. By a result from \cite{spm} a fixed point of SPM is a stair (each difference of height is 1) with at most one plateau (two consecutive columns with the same number of grain). As a consequence of those two results, there are at most three plateaus in $a$ (there may be one on the top, which we cut) at positions $(x,x+1)$ for the left plateau, $(y,y+1)$ for the top plateau and $(z,z+1)$ for the right plateau (see figure \ref{fig:plateaus}). We have $a \overset{*}{\triangleleft} b$ and $a \neq b$ so there exists at least one couple of positions $(i,j)$, with $i<j$, such that $a_i = b_i - 1$ and $a_j=b_j + 1$. Let us now see that we can't have more than one such couple of positions, which will prove that $\Delta(a,b) \in 0^*-\!\!10^*10^*$. There are 4 positions where we can remove a grain and still respect the plateaus requirement to be a PSSPM fixed point on: $x$, $y$, $y+1$ and $z+1$ (if there is no top plateau, we can still remove the top grain), and there are 2 positions where we can add a grain: $x+1$ and $z$. But if we add a grain at $z$, we have to remove a grain at a position greater than $z$ (recall that $a \overset{*}{\triangleleft} b$). The only possible candidate position is $z+1$, leading to a configuration which is not a fixed point since the difference of height between $z$ and $z+1$ becomes greater than 2. Therefore we can only add a grain at $x+1$. Now, where can we remove a grain: only on $z+1$ if there is a plateau at $(z,z+1)$ (otherwise there are two plateaus on the left or right side which is not a SPM fixed point), and only on the rightmost top column if there is no right plateau. This proves that $a \triangleleft b$ and $b$ is unique.
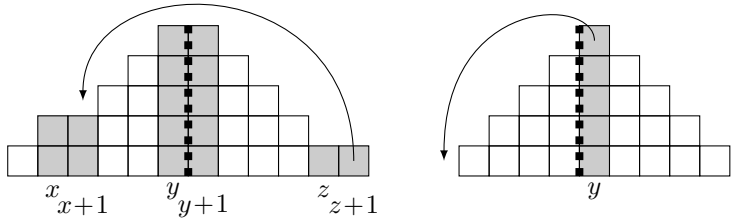
\begin{figure}[!h]
  \begin{center}
    \begin{tikzpicture}[scale=.4]
  \fill[black!20] (1,1) rectangle ++ (2,2);
  \fill[black!20] (5,1) rectangle ++ (2,5);
  \fill[black!20] (10,1) rectangle ++ (2,1);
  \foreach \x/\h in {0/1,1/2,2/2,3/3,4/4,5/5,6/5,7/4,8/3,9/2,10/1,11/1}
    \foreach \y in {1,...,\h}
      \draw(\x,\y) rectangle ++ (1,1);
  \node at (1.5,.5) {$x$};
  \node at (2.5,0) {$x\!+\!1$};
  \node at (5.5,.5) {$y$};
  \node at (6.5,0) {$y\!+\!1$};
  \node at (10.5,.5) {$z$};
  \node at (11.5,0) {$z\!+\!1$};
  \draw[line width=3pt,dashed] (6,1) -- ++ (0,5);
  \draw[-latex] (11.5,1.5) .. controls (11.5,8) and (2.5,8) .. (2.5,3.5);
  \fill[black!20] (15+4,1) rectangle ++ (1,5);
  \foreach \x/\h in {0/1,1/2,2/3,3/4,4/5,5/4,6/3,7/2,8/1}
    \foreach \y in {1,...,\h}
      \draw(15+\x,\y) rectangle ++ (1,1);
  \node at (15+4.5,.5) {$y$};
  \draw[line width=3pt,dashed] (15+4,1) -- ++ (0,5);
  \draw[-latex] (15+4.5,5.5) .. controls (15+4.5,7) and (15-.5,7) .. (15-.5,1.5);
\end{tikzpicture}
  \end{center}
  \caption{For any non-maximal fixed point $a$, there exists a unique fixed point $b$ such that $a \triangleleft b$.}
  \label{fig:plateaus}
\end{figure} 

\vspace{-15pt}
$~$
\end{proof}

\begin{theorem}\label{theorem:continuous}
Let
$$\pi_0 <_{lex} \pi_1 <_{lex} \dots <_{lex} \pi_{k-1} <_{lex} \pi_k$$
be the sequence of all fixed points of PSSPM($n$) ordered lexicographically. Then this sequence has the following strong relation:
$$\pi_0 \triangleleft \pi_1 \triangleleft \dots \triangleleft \pi_{k-1} \triangleleft \pi_k$$
Moreover, for any fixed point $\pi$ of SSPM($n$) such that $\pi_0 \leq_{lex} \pi \leq_{lex} \pi_k$, there exists an index $i$, $0 \leq i \leq k$, such that $\pi_i=\pi$.
\end{theorem}
\begin{proof}
  $n^2$ is an upper bound to the number of iterations from the configuration $(\underline{n})$ to a fixed point using PSSPM rules (at each step a grain loses some height). Therefore, the set of fixed points of PSSPM($n$) is equal to $\bigcup \limits_{\omega \in \{\mathcal L,\mathcal R\}^{n^2}} \omega((\underline{n}))$ because trying every possibility leads to reaching every possible fixed point.\\
  Starting from the word $s^0=\mathcal R^{n^2}$ and changing one by one the letters $\mathcal R$ into $\mathcal L$, we get a sequence of words $(s^0,s^1,\dots,s^{n^2})$ such that for all $k$, the size of the word $s^k$ is $n^2$ and the number of occurrences of $\mathcal L$ in $s^k$ is $k$. From Proposition \ref{lemma:main}, for all $k<n^2$ we have $s^k ((\underline{n})) \overset{*}{\triangleleft} s^{k+1}((\underline{n}))$. There are two possibilities:
  \begin{itemize}
    \item $s^k((\underline{n}))=s^{k+1}((\underline{n}))$.
    \item $s^k((\underline{n})) \neq s^{k+1}((\underline{n}))$.
  \end{itemize}
  In the second case, both configurations are fixed points of PSSPM($n$) and from Proposition \ref{lemma:psspm} we have $s^k \triangleleft s^{k+1}$. This gives a simple procedure to construct the set of fixed points of PSSPM($n$) from $\pi_0$ to $\pi_k$ and proves the first part of the theorem (the procedure is described below).

  From the SSPM($n$) fixed point characterization described in \cite{sspm} and \cite{unimodal} (presented in the proof of Proposition \ref{lemma:psspm}), even if the complete set of reachable fixed points are not the same, the fixed points of SSPM($n$) and PSSPM($n$) between the smallest and greatest fixed points of PSSPM($n$) are the same (the authors of \cite{sspm} and \cite{unimodal} use exactly the same construction as the one described in the proof of Proposition \ref{lemma:psspm}, see figure \ref{fig:plateaus}). The fact that PSSPM($n$) $\subseteq$ SSPM($n$) completes the proof of the second part of the theorem.
\end{proof}

The proofs of Proposition \ref{lemma:psspm} and Theorem \ref{theorem:continuous} provide a simple algorithm to browse the set of fixed points of PSSPM($n$). First compute the minimal (rightmost $\pi^{\mathcal R}$) and maximal (leftmost $\pi^{\mathcal L}$) fixed points starting from $(\underline{n})$ by following always the same choice ($\mathcal R$ to get the minimal configuration, and $\mathcal L$ to get the maximal one). Then starting from $\pi^{\mathcal R}$, construct the unique fixed point $\pi_1$ such that $\pi^{\mathcal R} \triangleleft \pi_1$, as explained on figure \ref{fig:plateaus}. From $\pi_1$, construct the unique fixed point $\pi_2$ such that $\pi_1 \triangleleft \pi_2$, etc... Until you get $\pi^{\mathcal L}$. From what precedes, this deterministic procedure browses exhaustively the set of fixed points of PSSPM($n$).


\section{Conclusion}

We have studied the set of fixed points of PSSPM($n$) and compared it to the set of fixed points of SSPM($n$) using the natural lexicographic order. We proved the intuitive fact that the greatest fixed point can be reached using always the choice $\mathcal L$, and that the smallest fixed point can be reached using always the choice $\mathcal R$. More interestingly, we showed that every fixed point reachable in SSPM($n$) between the lowest and the greatest fixed points of PSSPM($n$) is also reachable in PSSPM($n$). This is a kind of continuity property: the set of fixed points reachable in PSSPM($n$) is an "interval" of the set of fixed points reachable in SSPM($n$).

Further work may concentrate on finding a bound on the maximal and minimal non-empty columns in the set of fixed points of PSSPM($n$) which is an open question. The bound $\lfloor \sqrt{2n} \rfloor$ proved in \cite{unimodal} holds for PSSPM($n$) but it is not satisfying since proposition \ref{prop:psspm} states that there are strictly less fixed points in PSSPM($n$) than in SSPM($n$).\\

\subsection*{Acknowledgements}
The authors would like to thank Eric R\'emila for useful comments.

\bibliographystyle{alpha}
\bibliography{biblio}

\newcommand{\etalchar}[1]{$^{#1}$}
\begin{thebibliography}{RDMDP06}

\bibitem[BTW88]{soc}
P.~Bak, C.~Tang, and K.~Wiesenfeld.
\newblock Self-organized criticality.
\newblock {\em Phys. Rev. A}, 38(1):364--374, 1988.

\bibitem[CF03]{sa3}
Julien Cervelle and Enrico Formenti.
\newblock On sand automata.
\newblock In Helmut Alt and Michel Habib, editors, {\em STACS}, volume 2607 of
  {\em Lecture Notes in Computer Science}, pages 642--653. Springer, 2003.

\bibitem[CFM07]{sa2}
Julien Cervelle, Enrico Formenti, and Beno\^{\i}t Masson.
\newblock From sandpiles to sand automata.
\newblock {\em Theor. Comput. Sci.}, 381(1-3):1--28, 2007.

\bibitem[CK93]{spm}
Eric~Goles Ch. and Marcos~A. Kiwi.
\newblock Games on line graphs and sand piles.
\newblock {\em Theor. Comput. Sci.}, 115(2):321--349, 1993.

\bibitem[CLM{\etalchar{+}}04]{survey}
Eric~Goles Ch., Matthieu Latapy, Cl{\'e}mence Magnien, Michel Morvan, and
  Ha~Duong Phan.
\newblock Sandpile models and lattices: a comprehensive survey.
\newblock {\em Theor. Comput. Sci.}, 322(2):383--407, 2004.

\bibitem[CMP02]{order}
Eric~Goles Ch., Michel Morvan, and Ha~Duong Phan.
\newblock Sandpiles and order structure of integer partitions.
\newblock {\em Discrete Applied Mathematics}, 117(1-3):51--64, 2002.

\bibitem[DGM09]{sa1}
Alberto Dennunzio, Pierre Guillon, and Beno\^{\i}t Masson.
\newblock Sand automata as cellular automata.
\newblock {\em Theor. Comput. Sci.}, 410:3962--3974, September 2009.

\bibitem[DL98]{pspm}
J{\'e}r{\^o}me~Olivier Durand-Lose.
\newblock Parallel transient time of one-dimensional sand pile.
\newblock {\em Theor. Comput. Sci.}, 205(1-2):183--193, 1998.

\bibitem[FMP07]{sspm}
Enrico Formenti, Beno\^{\i}t Masson, and Theophilos Pisokas.
\newblock Advances in symmetric sandpiles.
\newblock {\em Fundam. Inform.}, 76(1-2):91--112, 2007.

\bibitem[FPPT10]{fppt}
E.~Formenti, V.~T. Pham, T.~H.~D. Phan, and T.~T.~H. Tran.
\newblock Fixed point form of the parallel symmetric sand pile model.
\newblock {\em preprint}, 2010.

\bibitem[LMMP01]{structure}
Matthieu Latapy, Roberto Mantaci, Michel Morvan, and Ha~Duong Phan.
\newblock Structure of some sand piles model.
\newblock {\em Theor. Comput. Sci.}, 262(1):525--556, 2001.

\bibitem[Mas09]{massa3}
Paolo Massazza.
\newblock A cat algorithm for sand piles.
\newblock {\em Pure Mathematics and Applications}, 19:147--158, 2009.

\bibitem[MM11]{massa1}
Roberto Mantaci and Paolo Massazza.
\newblock From linear partitions to parallelogram polyominoes.
\newblock In {\em Proceedings of the 15th international conference on
  Developments in language theory}, DLT'11, pages 350--361, Berlin, Heidelberg,
  2011. Springer-Verlag.

\bibitem[MR10]{massa2}
Paolo Massazza and Roberto Radicioni.
\newblock A cat algorithm for the exhaustive generation of ice piles.
\newblock {\em RAIRO - Theor. Inf. and Applic.}, 44(4):525--543, 2010.

\bibitem[Pha08]{unimodal}
Thi Ha~Duong Phan.
\newblock Two sided sand piles model and unimodal sequences.
\newblock {\em ITA}, 42(3):631--646, 2008.

\bibitem[RDMDP06]{bspm}
Dominique Rossin, Enrica Duchi, Roberto Mantaci, and Ha~Duong~Phan.
\newblock {Bidimensionnal sand pile and ice pile models}.
\newblock In {\em {GASCOM 2006}}, Dijon, France, 2006.

\end{thebibliography}

\end{document}